\documentclass{article}
\usepackage{fullpage}
\usepackage{xcolor}

\definecolor{pcgreen}{rgb}{0.0, 0.62, 0.24}

\usepackage{authblk}

\title{On Computing Pareto Optimal Paths in Weighted Time-Dependent Networks}

\author[1]{Filippo Brunelli}
\affil[1]{Inria, Université de Paris, CNRS, IRIF, F-75013 Paris, France (\texttt{filippo.brunelli@inria.fr})}

\author[2]{Pierluigi Crescenzi}
\affil[2]{Gran Sasso Science Institute, 67100 L'Aquila, Italy (\texttt{pierluigi.crescenzi@gssi.it})}

\author[3]{Laurent Viennot}
\affil[3]{Inria, Université de Paris, CNRS, IRIF, F-75013 Paris, France (\texttt{laurent.viennot@inria.fr})}

\usepackage{blindtext}
\usepackage{pgfplots}
\usepackage{tikz}
\usepackage{tikz-network}
\usetikzlibrary{snakes}
\usetikzlibrary{decorations.markings}
\usepackage[ruled,vlined,linesnumbered]{algorithm2e}
\usepackage{float}
\usepackage{mathtools}
\usepackage{tabularx}
\usepackage{multirow}
\usepackage{makecell}
\newcommand{\tdn}{Time-dependent network (TDN)}
\newcommand{\pltdn}{Piecewise linear TDN}
\newcommand{\cdtdn}{Constant-delay TDN}
\newcommand{\patdn}{Point-availability TDN}
\newcommand{\utdn}{Uniform TDN}
\newcommand{\fls}{Finite link stream}
\newcommand{\tn}{Temporal network}

\newcommand{\thor}{\mathbb{T}}
\newcommand{\cost}{\mathbb{C}}
\newcommand{\poset}[3]{\mathbb{PO}_{#1,#2}(#3)}

\newcommand{\tpath}[1]{\mathbb{#1}}
\newcommand{\allpaths}[2]{\mathcal{P}_{#1,#2}}

\newcommand{\sd}[2]{#1\textnormal{-}#2}
\newcommand{\tindedge}[1]{(u_{#1},v_{#1},\tau_{#1},\delta_{#1})}
\newcommand{\tedge}{(u,v,\tau,\delta)}

\newcommand{\problem}[2]{\medskip\fbox{\parbox{0.92\textwidth}{\textbf{#1.} #2}}\medskip}
\makeatletter
\newcommand*{\centerfloat}{%
  \parindent \z@
  \leftskip \z@ \@plus 1fil \@minus \textwidth
  \rightskip\leftskip
  \parfillskip \z@skip}
\makeatother

\usepackage{amsmath,amssymb}

\usepackage{soul}

\usepackage{blindtext}

\usepackage[ruled,vlined,linesnumbered]{algorithm2e}
\newcommand{\updateps}{\textsc{update\_ps}\xspace}

\newtheorem{theorem}{Theorem}
\newtheorem{lemma}{Lemma}

\def\Box{\hbox{\hskip 1pt \vrule width 4pt height 8pt depth 1.5pt \hskip 1pt}}
\newenvironment{proof}{\medskip\noindent\textbf{Proof.}}{{}\hfill$\Box$\\}

\begin{document}

\maketitle

\begin{abstract}
A weighted point-availability time-dependent network is a list of temporal edges, where each temporal edge has an appearing time value, a travel time value, and a cost value. In this paper we consider the single source Pareto problem in weighted point-availability time-dependent networks, which consists of computing, for any destination $d$, all Pareto optimal pairs $(t,c)$, where $t$ and $c$ are the arrival time and the cost of a path from $s$ to $d$, respectively (a pair $(t,c)$ is Pareto optimal if there is no path with arrival time smaller than $t$ and cost no worse than $c$ or arrival time no greater than $t$ and better cost). We design and analyse a general algorithm for solving this problem, whose time complexity is $O(M\log P)$, where $M$ is the number of temporal edges and $P$ is the maximum number of Pareto optimal pairs for each node of the network. This complexity significantly improves the time complexity of the previously known solution. Our algorithm can be used to solve several different minimum cost path problems in weighted point-availability time-dependent networks with a vast variety of cost definitions, and it can be easily modified in order to deal with the single destination Pareto problem. All our results apply to directed networks, but they can be easily adapted to undirected networks with no edges with zero travel time.
\end{abstract}

\section{Introduction}
\label{sec:introduction}

A \textit{time-dependent network} (in short, TDN) is a graph $G=(V,E)$ in which the delay or travel time of each edge changes over time~\cite{Cooke1966,Dreyfus1969,Orda1990,Casteigts2012}. Typically, the dependence on time of the delay is specified by associating to each edge $e=(u,v)\in E\subseteq V\times V$ a function $\alpha_e(t)$ which indicates, for each time $t$, the \textit{arrival time} in $v$ when the edge is traversed starting from $u$ at time $t$ (see, for example, the time-dependent network shown in Figure~\ref{fig:timedependentnetwork}). Note that, once the arrival time function is specified, the \textit{delay} (or travel time) of an edge $e$ at time $t$ can be easily computed as $\delta_e(t)=\alpha_e(t)-t$ (since we cannot yet travel back in time, this implies that $\alpha_e(t)$ has to be no smaller than $t$). Equivalently, arrival time functions can be easily obtained from delay functions. 
This general definition has been refined in several different ways in the last $30$ years, by assuming different properties of the edge arrival time functions. In particular, we can identify the following hierarchy of TDN models, where each model encompasses the next one.
\begin{description}
    \item[Piecewise linear TDN] For each edge $e$, the function $\alpha_e$ is piecewise linear, like, for example, the functions $\alpha_{e_2}(t)$ and $\alpha_{e_3}(t)$ in Figure~\ref{fig:timedependentnetwork}~\cite{Foschini2014}.
    
    \item[Constant-delay TDN] These are piecewise linear TDNs in which, for each edge, the slope of all linear segments of the corresponding arrival time function is equal to $1$, like, for example, the function $\alpha_{e_3}(t)$ in Figure~\ref{fig:timedependentnetwork}~\cite{Dehne2012}.

    \item[Point-availability TDN] These are constant-delay TDNs in which the domain of the arrival time functions is a finite subset $\thor$ of the set of real numbers~\cite{Wu2016}. A commonly used representation of a point-availability TDN simply consists in listing all the quadruples $\tedge$, such that the arrival time function of the edge $(u,v)$ at time $\tau$ is equal to $\tau+\delta$ (see the two commonly used visualizations of such representation shown in Figure~\ref{fig:pointavailabilitytdn}).

    \item[Uniform TDN] These are point-availability TDNs in which the delay is the same value $\delta$ for all edges and all time instants. For example, temporal graphs or networks~\cite{Michail2016,Crescenzi2019} are uniform TDNs in which $\delta=1$ and $\thor\subseteq\mathbb{N}$, while finite link streams~\cite{Latapy2018} are uniform TDNs in which $\delta=0$.
\end{description}

\begin{figure}
\centerfloat
\SetVertexStyle[FillColor=white]
\SetEdgeStyle[Color=black]
\begin{tikzpicture}[scale=0.8, every node/.style={scale=0.8}]
  \Vertex[x=0,y=0,label={$u_1$}]{1}
  \Vertex[x=4,y=0,label={$u_3$}]{3}
  \Vertex[x=2,y=2,label={$u_2$}]{2}
  \Edge[label={$e_1$},position={above left=0.5mm},Direct](1)(2)
  \Edge[label={$e_2$},position={below=0.5mm},Direct](1)(3)
  \Edge[label={$e_3$},position={above right=0.5mm},Direct](3)(2)
  \begin{scope}[xshift=5cm]
    \begin{axis}[axis x line=middle,axis y line=middle,grid=major,width=6cm,height=6cm,grid style={dashed, gray!30},
        xmin=0,xmax=5.5,ymin=0,ymax=9.5,xlabel=$t$,ylabel={$\alpha_{e_1}(t)$},tick align=outside,every axis x label/.style={at={(ticklabel* cs:1.05)},anchor=west},every axis y label/.style={at={(ticklabel* cs:1.05)},anchor=south}]
    \addplot[domain=0:5,dashed,samples=100] {x};
    \addplot[domain=0:5,ultra thick,samples=100] {-x*x+6*x};
    \end{axis}
  \end{scope}
  \begin{scope}[xshift=11cm]
    \begin{axis}[axis x line=middle,axis y line=middle,grid=major,width=6cm,height=6cm,grid style={dashed, gray!30},
        xmin=0,xmax=5.5,ymin=0,ymax=5.5,xlabel=$t$,ylabel={$\alpha_{e_2}(t)$},tick align=outside,enlargelimits=false,every axis x label/.style={at={(ticklabel* cs:1.05)},anchor=west},every axis y label/.style={at={(ticklabel* cs:1.05)},anchor=south}]
    \addplot[domain=0:5,dashed,samples=100] {x};
    \addplot[domain=0:1, ultra thick,samples=100] {2*x+1};
    \addplot[domain=1:2, ultra thick,samples=100] {0.75*x+1.25};
    \addplot[domain=3:5, ultra thick,samples=100] {0.5*x+2.5};
    \end{axis}
  \end{scope}
  \begin{scope}[xshift=17cm]
    \begin{axis}[axis x line=middle,axis y line=middle,grid=major,width=6cm,height=6cm,grid style={dashed, gray!30},
        xmin=0,xmax=5.5,ymin=0,ymax=5.5,xlabel=$t$,ylabel={$\alpha_{e_3}(t)$},tick align=outside,enlargelimits=false,every axis x label/.style={at={(ticklabel* cs:1.05)},anchor=west},every axis y label/.style={at={(ticklabel* cs:1.05)},anchor=south}]
    \addplot[domain=0:5,dashed,samples=100] {x};
    \addplot[domain=0:1, ultra thick,samples=100] {x+0.5};
    \addplot[domain=1:2, ultra thick,samples=100] {x+0.1};
    \addplot[domain=3:5, ultra thick,samples=100] {x+0.5};
    \end{axis}
  \end{scope}
\end{tikzpicture}
\caption{An example of time-dependent network. In this example, $\alpha_{e_1}(t)=-t^2+6t$. The function $\alpha_{e_2}(t)$ is piecewise linear (with segments $2t+1$ in $[0,1]$, $\frac{3}{4}t+\frac{5}{4}$ in $(1,2]$, and $\frac{1}{2}t+\frac{5}{2}$ in $[3,5]$), while the function $\alpha_{e_3}(t)$ is piecewise linear and constant-delay (with delay $\frac{1}{2}$ in $[0,1]$, $\frac{1}{10}$ in $(1,2]$, and $\frac{1}{2}$ in $[3,5]$).}
\label{fig:timedependentnetwork}
\end{figure}
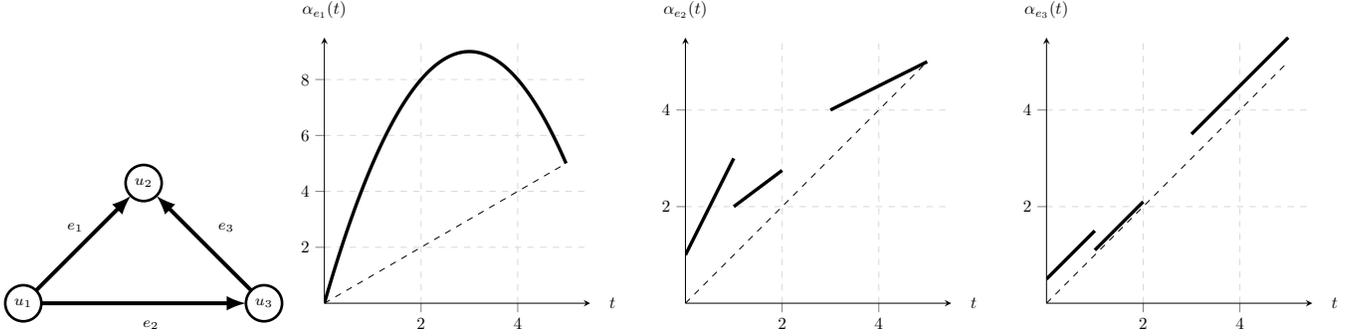

\begin{figure}[b]
\centerfloat
\SetVertexStyle[FillColor=white]
\SetEdgeStyle[Color=black]
\begin{tikzpicture}[x=5pt,y=5pt]
  \Vertex[x=0,y=0,label={$u_1$}]{1}
  \Vertex[x=5,y=0,label=$u_2$]{2}
  \Vertex[x=5,y=4,label=$u_3$]{3}
  \Vertex[x=0,y=4,label=$u_4$]{4}
  \Vertex[x=2.5,y=2,label=$u_5$]{5}
  \Edge[label={$9|4$},Direct](1)(2)
  \Edge[label={$1|2$},bend=30,Direct](1)(4)
  \Edge[label={$6|2$},Direct](1)(4)
  \Edge[label={$7|3$},Direct](1)(5)
  \Edge[label={$2|1$},Direct](2)(3)
  \Edge[label={$5|2$},Direct,bend=-20,distance=0.2](2)(4)
  \Edge[label={$11|1$},Direct](2)(5)
  \Edge[label={$9|2$},Direct](3)(4)
  \Edge[label={$3|1$},Direct](3)(5)
  \Edge[label={$12|1$},Direct,bend=-30](3)(5)
  \Edge[label={$3|1$},Direct](4)(5)
  \Edge[label={$8|1$},Direct,bend=-30](4)(5)
\end{tikzpicture}
\qquad
\begin{tikzpicture}[x=10pt,y=10pt]
\tikzset{vertex/.style={circle,fill=black,inner sep=0pt,minimum size=3pt},arc/.style={-,line width=.6pt}}
\node at (0,-4) {};
\draw [thick,dashed] (0,0) -- (26,0);
\draw [thick,dashed] (0,1) -- (26,1);
\draw [thick,dashed] (0,2) -- (26,2);
\draw [thick,dashed] (0,3) -- (26,3);
\draw [thick,dashed] (0,4) -- (26,4);
\node at (-1,4) {$u_1$};
\node at (-1,3) {$u_2$};
\node at (-1,2) {$u_3$};
\node at (-1,1) {$u_4$};
\node at (-1,0) {$u_5$};
\node at (1,5) {$1$};
\node at (3,5) {$2$};
\node at (5,5) {$3$};
\node at (7,5) {$4$};
\node at (9,5) {$5$};
\node at (11,5) {$6$};
\node at (13,5) {$7$};
\node at (15,5) {$8$};
\node at (17,5) {$9$};
\node at (19,5) {$10$};
\node at (21,5) {$11$};
\node at (23,5) {$12$};
\node at (25,5) {$13$};
\node[vertex] (e1u) at (1,4) {};
\node[vertex] (e1v) at (5,1) {};
\node[vertex] (e2u) at (3,3) {};
\node[vertex] (e2v) at (5,2) {};
\node[vertex] (e3u) at (5,1) {};
\node[vertex] (e3v) at (7,0) {};
\node[vertex] (e4u) at (7,2) {};
\node[vertex] (e4v) at (9,0) {};
\node[vertex] (e5u) at (9,3) {};
\node[vertex] (e5v) at (13,1) {};
\node[vertex] (e6u) at (11,4) {};
\node[vertex] (e6v) at (15,1) {};
\node[vertex] (e7u) at (13,4) {};
\node[vertex] (e7v) at (19,0) {};
\node[vertex] (e8u) at (15,1) {};
\node[vertex] (e8v) at (17,0) {};
\node[vertex] (e9u) at (17,2) {};
\node[vertex] (e9v) at (21,1) {};
\node[vertex] (e10u) at (17,4) {};
\node[vertex] (e10v) at (25,3) {};
\node[vertex] (e11u) at (21,3) {};
\node[vertex] (e11v) at (23,0) {};
\node[vertex] (e12u) at (23,2) {};
\node[vertex] (e12v) at (25,0) {};
\draw[arc,->] (e1u) -- (e1v);
\draw[arc,->] (e2u) -- (e2v);
\draw[arc,->] (e3u) -- (e3v);
\draw[arc,->] (e4u) -- (e4v);
\draw[arc,->] (e5u) -- (e5v);
\draw[arc,->] (e6u) -- (e6v);
\draw[arc,->] (e7u) -- (e7v);
\draw[arc,->] (e8u) -- (e8v);
\draw[arc,->] (e9u) -- (e9v);
\draw[arc,->] (e10u) -- (e10v);
\draw[arc,->] (e11u) -- (e11v);
\draw[arc,->] (e12u) -- (e12v);
\end{tikzpicture}
\caption{Two visualizations of an example of point-availability TDN. In the visualization on the left, each edge is labeled with its availibility time and its delay. In the second visualization, the delay of each edge can be computed as the difference between its starting time and its arrival time.}
\label{fig:pointavailabilitytdn}
\end{figure}
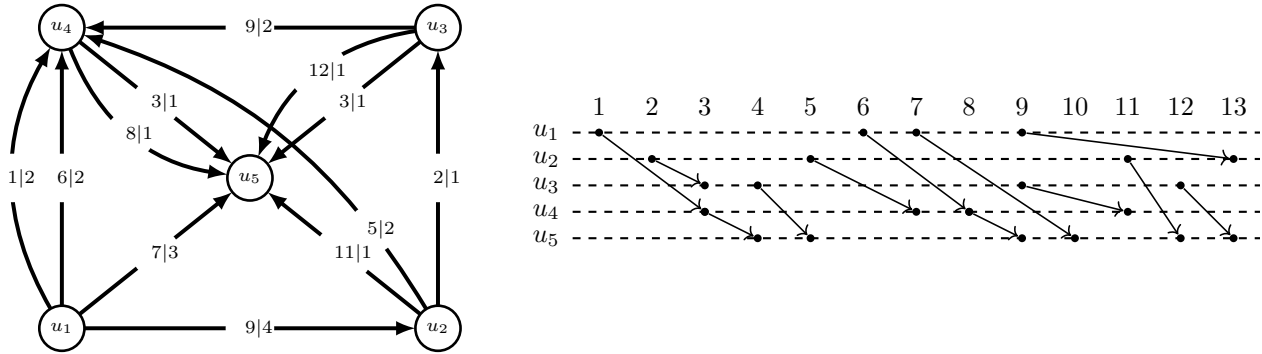

Besides the above constraints on the arrival time functions, it is also common practice to distinguish between TDNs which satisfy the FIFO property and TDNs which do not satisfy this property~\cite{Dean2004}. The \textit{FIFO property} states that, for every edge $e=(u,v)$, a later starting time at $u$ results in a later (or equal) arrival time at $v$. In other words, for each edge, the arrival time function of the edge itself is non-decreasing. Note that any uniform TDN satisfies this property, while, in general, this is not true for the other levels of the above hierarchy (see, for example, the arrival time functions of the TDN shown in Figure~\ref{fig:timedependentnetwork}).

One of the basic notions of TDNs is the definition of \textit{path}, which has to satisfy, besides the typical constraints of a path in a graph, some natural time constraints. In particular, a path $\tpath P=\langle (e_1,t_1),\ldots,(e_k,t_k) \rangle$ in a TDN, starting from a node $u$ at time $t_1$ and arriving to a node $v$ at time $t_{k+1}=\alpha_{e_k}(t_k)$, is such that each edge $e_i=(u_i,v_i)$ in the path is available at some associated time instant $t_i$ following the arrival time of the path in $u_i$ (in particular $t_0\le t_1$) and the arrival time in $v$ is no later than $t_{k+1}$. By referring to the TDN of Figure~\ref{fig:timedependentnetwork}, we have that, for example, $\langle (e_2,4),(e_3,\frac{9}{2})\rangle$ is a path from $u_1$ to $u_2$ starting at time $4$ and arriving at time $5$ (note that this path is faster than directly traversing the edge $e_1$ at time $4$, which would have taken to $u_2$ at time $8$).

Concerning the definition of paths, another distinction among TDNs is made by applying different waiting policies~\cite{Orda1990}. In this paper, we will focus on TDNs with an \textit{unrestricted waiting policy}, which allows us to wait at a node, as much as it is necessary, until an edge exiting from the node becomes available. For example, by referring to the TDN of Figure~\ref{fig:timedependentnetwork} (which does not satisfy the FIFO property), if we arrive at node $u_1$ at time $2.5$, we can wait until time $5$ (respectively, $3$) in order to traverse the edge $e_1$ (respectively, $e_2$): in both cases, the arrival time at the other extreme of the edge would be smaller than traversing the edge at time $2.5$ (actually, in the case of $e_2$, the edge is not even available at time $2.5$).

A very well-studied problem on TDNs is the so-called (single-source) \textit{profile problem}, that is, \textit{given a TDN and given a source $s$, compute, for each destination $u$, its profile function which associates to any time $t$, the earliest arrival time in $u$, if we start from $s$ at time $t$}. The above hierarchy of TDNs (see the left part of Figure~\ref{fig:hierarchy}) induces different complexities for the profile problem. In order to describe these differences, let us first introduce some complexity parameters. The size of a TDN can be expressed in terms of the number $n=|V|$ of nodes, the number  $m=|E|$ of edges, and the sum of arrival time function sizes $S=\sum_{e\in E}|\alpha_e|$, where the \textit{size} $|\alpha_e|$ of $\alpha_e$ is the number of parameters required to store a representation of $\alpha_e$.

\begin{figure}
\centerfloat
\raisebox{-.5\height}{\tikzset{>=stealth',did/.style={rectangle,draw=none,fill=none,inner sep=0cm}}
\begin{tikzpicture}[scale=1,every node/.style={draw, rectangle, align=center,font=\footnotesize},label/.style={did, inner sep=1pt,midway,above,sloped,font=\tiny}]
\node[fill=white] (tdn) at (0,4) {\tdn};
\node[fill=white] (pltdn) at (0,3) {\pltdn};
\node[fill=white] (cdtdn) at (0,2) {\cdtdn};
\node[fill=white] (patdn) at (0,1) {\patdn};
\node[fill=white] (utdn) at (0,0) {\utdn};
\node[fill=white] (fls) at (-1.5,-1) {\fls};
\node[fill=white] (tn) at (1.5,-1) {\tn};
\draw[->, thick] (fls) -- node[draw=none,left] {$\delta=0$} (utdn);
\draw[->, thick] (tn) -- node[draw=none,right] {$\delta=1$} (utdn);
\draw[->, thick] (utdn) -- node[draw=none,right] {$\alpha_e(t)-t=\delta$ for each edge} (patdn);
\draw[->, thick] (patdn) -- node[draw=none,right] {$|\mathbb{T}|<\infty$} (cdtdn);
\draw[->, thick] (cdtdn) -- node[draw=none,right] {$\alpha_e(t)-t$ constant in each interval} (pltdn);
\draw[->, thick] (pltdn) -- node[draw=none,right] {$\alpha_e(t)$ piecewise linear} (tdn);
\end{tikzpicture}}\qquad\begin{tabular}{||c||>{\centering}m{2.5cm}|>{\centering\arraybackslash}m{3cm}||}
\hline
\multirow{2}{*}{\textbf{TDN class}} & \multicolumn{2}{c||}{\textbf{Complexity profile problem}}\\
\cline{2-3}
 & \textit{Time} & \textit{Profile size $P$}\\
\hline
\makecell{Piecewise-\\ linear} & \makecell{$O(nPm)$\\ \cite{Orda1990}}
& \makecell{$P=S n^{O(\log n)}$\\ can be $n^{\Omega(\log n)}$ \cite{Foschini2014}}\\
\hline
\makecell{Constant-\\ delay} & \makecell{$O(S(m+n\log n))$\\ \cite{Dehne2012}} & \\
\cline{1-2}
\makecell{Point-\\ availability} & \makecell{$O(S\log P)$\\ implicit in\\ \cite{Dibbelt2013,Dibbelt2018,Wu2014Path,Wu2016}} & {$P=O(S)$}\\
\cline{1-2}
\makecell{Temporal\\ network} & \makecell{$O(S)$\\ implicit in\\ \cite{Kossinets2008Structure}} & \\
\hline
\hline
\end{tabular}

\caption{The hierarchy of time-dependent networks (left) and the corresponding complexities of the profile problem (right). We use $n=|V|$, $m=|E|$, $S=\sum_{e\in E}|\alpha_e|$, $|\alpha_e|$ is the number of parameters required to store a representation of $\alpha_e$, and $P$ denotes the maximum size of a profile.}
\label{fig:hierarchy}
\end{figure}

The complexity of the profile problem increases as we go higher in the hierarchy, as described in the right part of Figure~\ref{fig:hierarchy} for non-zero delays and unrestricted waiting policy. The profile problem was first studied in general TDNs in~\cite{Orda1990} with an algorithm whose complexity depends on the size of the arrival time functions used to represent profiles. In the piecewise-linear case, each profile function is also piecewise linear and its size can be defined similarly to arrival time functions. Unfortunately the maximum size $P$ of a profile function can be super-polynomial as shown in~\cite{Foschini2014} in the case of piecewise-linear TDNs. A first gap occurs for constant-delay TDNs where profile size is $O(S)$~\cite{Dehne2012} and the best algorithm (as far as we know) is quadratic~\cite{Dehne2012}. A second important gap occurs for point-availability TDNs. Surprisingly, it has received little attention in the literature and we could not find any work explicitly stating its complexity. The profile version of the CSA algorithm~\cite{Dibbelt2013,Dibbelt2018} solves the problem in a public transit network model which is more sophisticated and for which the complexity is not stated. The algorithm consists in a single scan of what we define as temporal edges later and an $O(S\log P)$ complexity can easily be inferred. A similar algorithm is proposed in~\cite{Wu2014Path,Wu2016} for fastest path computation and could be easily transformed into a profile algorithm. A linear time algorithm can be inferred from the vector clock algorithm in \cite{Kossinets2008Structure} for temporal networks. 


Interestingly, the profile problem associated to a specific source $s$ can be seen as a bi-criteria path problem when considering both the starting and the arrival time. Indeed, the profile function of a destination $u$ can be seen as a set of Pareto pairs $(a_t,-t)$ such that $a_t$ is the earliest arrival time in $u$ starting from $s$ at time $t$ (see Section~\ref{sec:applications} for a precise formulation of this statement). Inspired by this observation and by the rich literature on multi-criteria path problems in directed networks, which started at least at the beginning of the eighties~(see, for example, \cite{Hansen1980,Martins1984}), in this paper we extend the definition of a TDN in order to deal with a multiplicity of objectives while computing paths starting from a given source. To this aim, we integrate the definition of a TDN with an edge \textit{cost function}, a \textit{cost combination function}, and a \textit{cost total order}. The edge costs combine along a path according to the cost combination function, that is, the cost a path $\tpath P=\langle e_1,\ldots,e_k\rangle$ is equal to the cost of the sub-path $\langle e_1,\ldots,e_{k-1}\rangle$ combined with the cost of the edge $e_k$ (and by taking into account the arrival times). 

The problem we focus on in this paper is then the following one. We want to compute the set of Pareto optimal values of the paths from a given source starting at a given time to all the possible destinations, with respect to two criteria:  arrival time and cost. More precisely, given a time $t_0$ and two nodes $s$ and $d$, a path $\tpath P$, starting from $s$ at time $t_0$ and arriving to $d$, is \emph{Pareto $t_0$-optimal}, if there is no path $\tpath Q$ starting from $s$ at time no earlier than $t_0$ and arriving to $d$ such that the arrival time of $\tpath Q$ is smaller than the arrival time of $\tpath P$ and its cost is not greater (according to the cost total order) or the arrival time of $\tpath Q$ is no greater than the arrival time of $\tpath P$ and its cost is smaller. The (single-source) \textit{Pareto problem} can then be defined as follows: \textit{given a TDN, a cost function along with its cost combination function and its cost total order, a source node $s$, and a starting time $t_0$, compute, for each destination $d$, all pairs $(t,c)$ for which there exists a Pareto $t_0$-optimal path, starting from $s$ at time no earlier than $t_0$ and arriving to $d$, whose arrival time is $t$ and whose cost is $c$}.

This problem was implicitly considered in \cite{Mutzel2019} where the enumeration of all Pareto optimal paths in a point-availability TDN is considered. The first phase of their enumeration algorithm for min-cost earliest arrival paths solves the Pareto problem, as we stated it here, in $O(S^2)$ time. We improve over their algorithm in two ways. First, we obtain a much lower time complexity, that is, $O(S\log P)$, and second we identify the key algebraic property that allows us to generalize the framework to a large variety of cost definitions. The Pareto problem indeed appears to be a corner-stone problem for solving various minimum costs problems in TDNs.
Some specific cases of this problem corresponding to specific edge cost functions have been considered in the literature~\cite{Orda1991,Xuan2003,Simard2019}. More precisely, we show how the Pareto problem can be solved in time $O(S\log P)$, whenever the cost functions satisfy a very natural property, called  \textit{isotonicity}, which is similar to the one used in~\cite{Sobrinho2005} while developing an algebraic approach for path-vector routing. Our main contributions are then the following (our results hold in the case in which there are no edges with delay equal to $0$ or the set of edges that have the same departure time and delay equal to $0$ do not induce any loop).

\begin{enumerate}
    \item If a cost function satisfies the isotonicity property, then the Pareto problem can be solved in time $O(S\log K)$, and in space $O(S)$ where $K$ is a Pareto complexity parameter satisfying $K\le P$ (see Section~\ref{sec:algorithm}). Hence, despite its generality, the Pareto problem can be solved, in the case of point-availability TDNs, with the same complexity of the profile problem, which is a special case  of the Pareto problem (see below).
    The first phase of the path-enumeration algorithm for min-cost earliest arrival paths proposed in~\cite{Mutzel2019} solves the Pareto problem in $O(S^2)$ time. This is the only and best complexity obtained prior to this paper as far as we know.
    
    \item The following path problems can be solved with the same time and space complexity.
    \begin{enumerate}
        \item Single source profile problem. As we already observed, the profile problem can be seen as a Pareto problem, where the cost function of an edge is the opposite of the departure time of the edge itself.
        We thus make explicitly its $O(S\log P)$ complexity while it was implicit in \cite{Dibbelt2013,Dibbelt2018,Wu2014Path,Wu2016,Kossinets2008Structure}. In the case of temporal networks (uniform delay 1 and integral times) the complexity is linear as $K=1$ in that case.
        
        \item Fewest hops. Finding the minimum number of edges required to reach each possible destination from a given source node. If $D=O(n)$ denotes the hop diameter (i.e. the maximum number of edges to reach a destination), a $O(Dm\log \max_{e\in E}|\alpha_e|)$-time algorithm was proposed in the context of time evolving graphs~\cite{Xuan2003} which can be seen as a particular case of constant-delay TDNs. The $O(S\log K)$ complexity we obtain for point-availability TDNs is better when the number of nodes is larger than the average number of events per edge which is the case in most practical networks.
        
        \item Shortest delay. Computing the paths with minimum duration defined as the sum of the delay of the edges.
        This problem was solved in~\cite{Wu2016} with similar $O(S\log P)$ complexity.
        
        \item Shortest fastest. Computing the paths with fewest hops among the paths with minimal total duration, defined as the difference between arriving and starting times of the path.
        This problem was introduced in the context of link streams~\cite{Latapy2018} that can be seen as constant-delay TDNs with uniform delay 0. An $O(n^2S^2\log S)$-time algorithm is proposed in \cite{Simard2019}. A variation of the algorithm is also proposed with similar complexity as our framework in a restricted model close to temporal networks. This variation does not seem to extend to point-availability TDNs.
        

        \item Min/Max of Sum/Product/Min/Max general minimum cost paths. Our technique applies to many costs considered previously. In particular, all the single criteria considered in \cite{Hansen1980} for searching min or max cost when combining costs within a path with Sum, Product, Min and Max are covered by our approach (as long as costs are positive in the case of Product). Note that fewest hops and shortest delay correspond to MinSum where costs are 1 and delay of edges respectively.
        If edges are associated to a reliability corresponding to the probability of not failing, then a path with highest reliability corresponds to MaxProduct (assuming independence of edge failures). As another example, if edges correspond to road segments and their cost is their steepness, then MinMax corresponds to a path that encounters the least steepness which can be interesting for bike route planning.
    \end{enumerate}
\end{enumerate}

All our results apply to directed point-availability TDNs. However, they can be easily adapted to undirected TDNs with no edges with zero delay.




\section{Definitions}
\label{sec:definitions}

As we said in the introduction, in this paper we focus our attention on point-availability TDNs. A \textit{point-availability time-dependent network} (in short, \textit{PATDN}) is a pair $\mathbb{G}=(V,\mathbb{E})$, where $V$ is the set of $n$ \textit{nodes} and $\mathbb{E}$ is the set of temporal edges. A \textit{temporal edge} $e$ is a quadruple $\tedge$, where $u\in V$ is the \textit{tail} of $e$, $v\in V$ is the \textit{head} of $e$, $\tau\in\mathbb{R}$ is the \textit{appearing time} of $e$, and $\delta\in\mathbb{R}$ is the \textit{delay} of $e$ (in the following, we will also refer to the \textit{arrival time} of $e$ defined as $\tau+\delta$). Note that the size of the network is $S=4|\mathbb{E}|$. Note also that this definition is slightly more general than the one given in the introduction, since we also allow the TDN to include edges with the same head, tail, and appearing time, but with different delay. We will allow temporal edges to have delay equal to $0$, but we will require that the set $\mathbb{E}_t$ of temporal edges that have the same appearing time $t$ and delay equal to $0$ do not induce any loop, that is, the \textit{zero $t$-snapshot graph} $G_t=(V,E_t)$ is a directed acyclic graph, where $(u,v)\in E_t$ if and only if $(u,v,t,0)\in\mathbb{E}_t$. Note that this property implies that the set $\mathbb{E}_t$ can be topologically ordered. In Section~\ref{subsec:zero}, we will see how we can relax this property in order to deal with more general cases.
Given a PATDN $\mathbb{G}=(V,\mathbb{E})$ and two nodes $u,v\in V$, a \textit{$\sd uv$ path} $\tpath P$ from $u$ to $v$  is a sequence of temporal edges $\langle e_1=\tindedge1,\ldots,e_k=\tindedge k\rangle\subseteq \mathbb{E}^k$ such that $u=u_1$, $v=v_k$, and, for each $i$ with $1<i\leq k$, $u_i=v_{i-1}$ and $\tau_i\geq\tau_{i-1}+\delta_{i-1}$. The \textit{starting time} of $\tpath P$ is defined as $\tau_1$, while the \textit{arrival time} $\alpha_{\tpath P}$ of $\tpath P$ is defined as $\tau_k+\delta_k$. As we said in the introduction, in this paper we deal with the \textit{unrestricted waiting} traversal policy, according to which it is possible to wait at a node as much as we want until some temporal edge appears and allows us to leave the node.

We now extend the definition of a PATDN in order to deal with a multiplicity of objectives while computing paths starting from a given source. To this aim, we integrate a PATDN $\mathbb{G}=(V,\mathbb{E})$ with a \textit{cost structure} $\mathcal{C}=(\cost,\gamma,\oplus\preceq)$ over $\mathbb{E}$, where $\cost$ is the set of possible \textit{cost values}, $\gamma$ is a \textit{cost function} $\gamma:\mathbb{E}\rightarrow\cost$, $\oplus$ is a \textit{cost combination function} $\oplus:\cost\times\cost\rightarrow\cost$, and $\preceq$ is a \textit{cost total order} $\preceq\ \subseteq\ \cost\times\cost$. For any path $\tpath P=\langle e_1,\ldots,e_k\rangle$, the \textit{cost function} of $\tpath P$ is recursively defined as follows: $\gamma_{\tpath P} = \gamma_{\langle e_1,e_2,\ldots,e_{k-1}\rangle} \oplus\gamma(e_k)$, with $\gamma_{\langle e_1\rangle}=\gamma(e_1)$ (in other words, the costs combine along the path according to the cost combination function).
Given a PATDN $\mathbb{G}=(V,\mathbb{E})$ with a cost structure $\mathcal{C}=(\cost,\gamma,\oplus\preceq)$ over $\mathbb{E}$, let $\thor$ denote the set of all real values $t$ such that there exists at least one temporal edge in $\mathbb{E}$ with appearing time or arrival time equal to $t$. We say that a pair $(t_1,c_1)\in\thor\times\cost$ \textit{dominates} a pair $(t_2,c_2)\in\thor\times\cost$ if $t_1 < t_2$ and $c_1\preceq c_2$, or $t_1\leq t_2$ and $c_1\prec c_2$ (the relation $\prec$ between the elements of $\cost$ is defined as $a \prec b$ if and only if $a \preceq b$ and $a \neq b$). Moreover, for any two nodes $u,v\in V$ and for any $t\in\mathbb{R}$, let $\allpaths uv(t)$ denotes the set of all $\sd uv$ paths whose starting time is no smaller than $t$.  Given a time $t_0\in\thor$, a path $\tpath P\in\allpaths sd(t_0)$ is \emph{Pareto $t_0$-optimal} among all paths in $\allpaths sd(t_0)$, if there is no path $Q\in\allpaths sd(t_0)$ such that $(\alpha_{\tpath Q},\gamma_{\tpath Q})$ dominates $(\alpha_{\tpath P},\gamma_{\tpath P})$. The problem we focus on in the rest of the paper is then the following one.

\problem{The Pareto problem}{Given a PATDN $N=(V,\mathbb{E})$ with a cost structure $\mathcal{C}=(\cost,\gamma,\oplus\preceq)$ over $\mathbb{E}$, a source node $s\in V$, and a starting time $t_0\in\mathbb{R}$, compute, for each destination $d\in V$, the set $\poset s{t_0}d$ containing all pairs $(t,c)\in\thor\times\cost$ for which there exists a Pareto $t_0$-optimal path $\tpath P\in\allpaths sd(t_0)$ such that $t=\alpha_{\tpath P}$ and $c=\gamma_{\tpath P}$.}

\section{Solving the Pareto problem in PATDNs}
\label{sec:algorithm}

Similarly to what has been observed in~\cite{Batz2012,Wu2016}, the prefix of a Pareto $t_0$-optimal path is not necessarily a Pareto $t_0$-optimal path. In order to deal with this problem, we assume that the cost structure $\mathcal{C}=(\cost,\gamma,\oplus\preceq)$ satisfies the following property.

\problem{Isotonicity property}{Let $c_1, c_2 \in\cost$ such that $c_1 \preceq c_2$. Then $c_1 \oplus c \preceq c_2 \oplus c$ for any $c \in \cost$.}

This property guarantees that, for any two paths $\tpath P_1$ and $\tpath P_2$ such that $\gamma_{\tpath P_1}\preceq\gamma_{\tpath P_2}$, and for each temporal edge $e$ that can be concatenated to both the paths $\tpath P_1$ and $\tpath P_2$, the cost of $\tpath P_2$ concatenated with $e$ is no better than the cost of $\tpath P_1$ concatenated with $e$. The isotonicity property also allows us to state the following lemma (in the following, two paths in $\tpath P_1$ and $\tpath P_2$ are said to be \textit{equivalent} if $\alpha_{\tpath P_1}=\alpha_{\tpath P_2}$ and $\gamma_{\tpath P_1}=\gamma_{\tpath P_2}$).

\begin{lemma}\label{lem:prefopt}
Let $N=(V,\mathbb{E})$ be a PATDN with a cost structure $\mathcal{C}=(\cost,\gamma,\oplus\preceq)$ over $\mathbb{E}$ satisfying the isotonicity property, let $s,d\in V$, and let $t_0\in\mathbb{R}$. Given a Pareto $t_0$-optimal path $\tpath P=\langle e_1,e_2, \dots,e_k\rangle\in\allpaths sd(t_0)$ with $k\geq2$, there exists a path $\tpath P' = \langle e_1',e_2',\dots,e_h'\rangle\in\allpaths sd(t_0)$ equivalent to $\tpath P$, such that its prefix $\langle e_1',e_2',\dots,e_{h-1}'\rangle$ is a Pareto $t_0$-optimal path. 
\end{lemma}
\begin{proof}
    Let $e_k = \tedge$ be the last temporal edge of $\tpath P$. Among all the paths $\tpath Q\in\allpaths su(t_0)$ with $\alpha_{\tpath Q}\leq \tau$, let us consider the set $\cal{M}$ of paths with minimum cost: hence, for any $\tpath M\in\cal M$, we have that $\gamma_{\tpath M}\preceq\gamma_{\langle e_1,e_2, \dots,e_{k-1}\rangle}$. Among the paths in $\cal{M}$, let $\tpath Q'=\langle e_1',e_2',\dots,e_{h'-1}\rangle$ be one with minimum arrival time. Note that $\tpath Q'$ is Pareto $t_0$-optimal among all paths in $\allpaths su(t_0)$. We define $\tpath P'$ as $\tpath Q'$ concatenated with $e_k$, that is, $\tpath P'=\langle e_1',e_2',\dots,e_{h'-1},e_k\rangle$. Clearly, $\alpha_{\tpath P'}=\tau+\delta=\alpha_{\tpath P}$. Moreover, since $\tpath P$ is Pareto $t_0$-optimal among all paths in $\allpaths sd(t_0)$, we have that $\gamma_{\tpath P}\preceq\gamma_{\tpath P'}$. Finally, by the isotonicity property, it follows that $\gamma_{\tpath P'}\preceq\gamma_{\tpath P}$. Hence, $\gamma_{\tpath P'}=\gamma_{\tpath P}$, which implies that $\tpath P$ and $\tpath P'$ are equivalent. Since the prefix of $\tpath P'$ is $\tpath Q'$ which is Pareto $t_0$-optimal, the lemma follows.
\end{proof}

We are now ready to describe our algorithm solving the Pareto problem, when the cost function satisfies the isotonicity property (see Algorithm~\ref{alg:paretoproblem}). To this aim, we assume that the set $\mathbb{E}$ of temporal edges with appearing time at least $t_0$ is ordered by increasing arrival time, prioritizing temporal edges with delay greater than $0$, and then topologically sorting the temporal edges that have delay $0$. At the beginning of the algorithm execution, all the Pareto sets are set to empty. For each scanned temporal edge $e$ in the PATDN, the algorithm updates the Pareto set of the head of $e$ by simply considering the cost of $e$ (if the tail of $e$ is the source node), and then updates the Pareto set of the head of $e$ by combining the cost of $e$ with the cost corresponding to the pair in the current Pareto set of the tail of $e$ with the greatest arrival time before the departure time of $e$. The update operation, with input a pair $(t,c)$, either simply adds the pair to the Pareto set (if this set is empty), or checks whether $c$ is smaller than the cost in the pair with greatest arrival time (recall that the temporal edges are sorted by increasing arrival time). In this latter case, it either simply adds the pair $(t,c)$ to the Pareto set (if the greatest arrival time is not equal to $t$) or substitutes the pair with greatest arrival time with the pair $(t,c)$ (since this latter pair dominates the one with greatest arrival time).

\begin{algorithm}[t]
\SetAlgoLined
\SetKwInOut{Input}{input}
\Input{An instance of the Pareto problem, in which the temporal edges of the PATDN are sorted as specified in the text}
\lForEach{$u \in V$}{
    $\poset s{t_0}u \gets \emptyset$
}
\ForEach{$e=\tedge$}{
    \lIf{$u = s$}{
        $\updateps (\poset s{t_0}v,\tau +\delta,\gamma(e))$ \label{alg:case1}
    }
    \If{$\poset s{t_0}u$ contains a pair $(t,c)$ with $t\leq\tau$}{\label{alg:case2a}
        let $(t^u, c^u)$ be the pair in $\poset s{t_0}u$ with greatest arrival time $t^u\leq\tau$\;\label{alg:costly}
        $\updateps (\poset s{t_0}v, \tau + \delta, c^u \oplus \gamma(e))$\;\label{alg:update2}
    }\label{alg:case2b}
}
\SetKwFunction{UpdatePS}{\updateps}
\SetKwProg{Fn}{Function}{:}{}
\Fn{\UpdatePS{$\mathbb{PO},t,c$}}{
    \lIf{$\mathbb{PO} = \emptyset$}{
        append $(t,c)$ to $\mathbb{PO}$
    }\Else{
        let $(t^*,c^*)$ be the pair in $\mathbb{PO}$ with greatest arrival time\;
        \If{$c \prec c^*$}{
            \If{$t = t^*$}{
                remove $(t^*,c^*)$ from $\mathbb{PO}$\;\label{alg:removal}
            }
            append $(t,c)$ to $\mathbb{PO}$\;
        }
    }
}
\caption{Computing, for each node $u$, the set $\poset s{t_0}u$}
\label{alg:paretoproblem}
\end{algorithm}

\begin{theorem}\label{thm:correctness}
For any instance of the Pareto problem, Algorithm~\ref{alg:paretoproblem} correctly computes, for any node $u$, the set $\poset s{t_0}u$.
\end{theorem}

\begin{proof}
Let $\mathbb{E}_k\subseteq\mathbb{E}$ be the set of the temporal edges scanned after $k$ temporal edges have been scanned. We will prove, by induction on $k$, the following property: the Pareto sets after $k$ temporal edges have been scanned represent the solution to the instance of the Pareto problem in which the set of temporal edges of the PATDN is restricted to $\mathbb{E}_k$. The correctness of the algorithm will follow by taking $k=|\mathbb{E}|$.

The property is clearly satisfied for $k=0$: indeed, the Pareto sets are all empty because there are no paths at all when $\mathbb{E}_0 = \emptyset$, that is, there is no temporal edge. Now suppose that the property holds for $k\geq0$ and let us prove it for $k+1$. Let $e =\tedge$ be the $k+1$-th scanned temporal edge. Let us distinguish the following two cases.
\begin{itemize}
    \item $\delta>0$: since the temporal edges are ordered by increasing arrival time, there cannot be temporal edges in $\mathbb{E}_{k+1}$ that depart from $v$ at time $\tau+\delta$ or later. This because, if a temporal edge $f$ departs at time $\tau+\delta$ or later, then $f$ has greater arrival time than $e$, except in the case $f$ has delay $0$ and departs at time $\tau+\delta$. Because of the used order of the temporal edges, this case is not possible (among the temporal edges with same arrival time, the ones with delay greater than $0$ are scanned before). Thus the new paths created by using $e$ must have $e$ as their last temporal edge.
    \item $\delta=0$: similarly to the previous case, there cannot be temporal edges in $\mathbb{E}_{k+1}$ that depart from $v$ later than $\tau+\delta$. However, $\mathbb{E}_{k+1}$ can contain temporal edges that depart at time $\tau+\delta$ and have delay equal to $0$. Since we are scanning the temporal edges with the same arrival time and length equal to $0$ by their topological order, none of these temporal edges departs from $v$. Thus, again, the new paths created by using $e$ have $e$ as their last temporal edge.
\end{itemize}
In summary, the only Pareto set that might be updated while scanning the temporal edge $e$ is the Pareto set of $v$. Note that when we update a Pareto set by considering a new pair $(t,c)$, this pair cannot dominate any pair already present in the Pareto set but the one with greatest arrival time: indeed, all the other pairs have earlier arrival times, and thus cannot be dominated by $(t,c)$. Let $(t^*,c^*)$ be the pair in the Pareto set with greatest arrival time. The pair $(t,c)$ dominates the pair $(t^*,c^*)$ only if $t=t^*$ and $c\prec c^*$: in this case we need to remove $(t^*,c^*)$ from the Pareto set and add $(t,c)$. Whenever $t>t^*$ but $c\prec c^*$, we have that neither $(t^*,c^*)$ dominates $(t,c)$ nor $(t,c)$ dominates $(t^*,c^*)$: in this case, the pair $(t,c)$ has to be added to the Pareto set (without removing the pair $(t^*,c^*)$). This is exactly what is done by the function \updateps\ in Algorithm~\ref{alg:paretoproblem}, where we also consider the case in which the Pareto set is empty (in this case, the pair $(t,c)$ is simply added to the Pareto set).

Let us now consider first the case in which the temporal edge $e$ is starting a new path $\tpath P=\langle e\rangle$, that is, the tail of $e$ is the source $s$. In this case, the pair $(\tau+\delta,\gamma(e))$ is a potential candidate to become a member of the Pareto set of $v$. For this reason, Algorithm~\ref{alg:paretoproblem} at line~\ref{alg:case1} invokes the function \updateps\ with arguments the Pareto set of $v$, $\tau+\delta$, and $\gamma(e)$.

It remains to consider the case in which $e$ extends a previous Pareto $t_0$-optimal path from $s$ to $u$. By the induction hypothesis, all the pairs corresponding to the Pareto $t_0$-optimal paths from $s$ to $u$, that can be concatenated with $e$, are already included in the Pareto set of $u$. Because of Lemma \ref{lem:prefopt}, in order to update the Pareto set of $v$, we just need to examine the Pareto set of $u$. Moreover, we are just interested in the paths that arrive in $u$ no later than $\tau$, so that adding $e$ to any of these paths produces new valid paths. Among those paths, we only have to consider the one with lowest cost, since adding $e$ to them will always produce a path arriving at the same time $\tau+\delta$. Let $\tpath P\in\allpaths su(t_0)$ be such a path, and let $t^u$ and $c^u$ be its arrival time and its cost, respectively. Note that, as we consider only Pareto $t_0$-optimal paths, $\tpath P$ is also the one having greatest arrival time among those arriving no later than $\tau$. The isotonicity property guarantees that $\tpath P$ produces the path with better cost after concatenation with the temporal edge $e$. Using the function $\oplus$, we can compute the cost $c^u\oplus\gamma(e)$ of $\tpath P$ concatenated with $e$: the pair $(\tau+\delta,c^u\oplus\gamma(e))$ is then a potential candidate to become a member of the Pareto set of $v$. This process is exactly what is done by Algorithm~\ref{alg:paretoproblem} at lines~\ref{alg:case2a}-\ref{alg:case2b}.

We have thus proved that Algorithm~\ref{alg:paretoproblem} correctly updates the Pareto set of $v$ when analysing the $k+1$-th scanned temporal edge, and thus proved the inductive step. The theorem thus follows.
\end{proof}

In order to analyse the complexity of Algorithm~\ref{alg:paretoproblem}, we introduce the following parameters of a PATDN $N$. For each temporal edge $e=\tedge$, the \textit{Pareto complexity} $K_e$ of $e$ is defined as the number of pairs $(t,c)$ in $\poset s{t_0}u$ such that $t\in(\tau,\tau+\delta]$. The Pareto complexity of $N$ is defined as $K = \max_{e \in \mathbb{E}}K_e$.

\begin{theorem}\label{thm:complexity}
With input any instance of the Pareto problem, Algorithm~\ref{alg:paretoproblem} executes in time $O(|\mathbb{E}|\log K)$ and in space $O(|\mathbb{E}|)$.
\end{theorem}

\begin{proof}
We assume that the temporal edges are already ordered as described above, and that the $\oplus$ and $\preceq$ operations require constant time. Moreover, note that the removal operations, that are possibly executed at line~\ref{alg:removal} of Algorithm~\ref{alg:paretoproblem}, require only to remove the last element of a list (since it turns out that the lists are ordered with respect to the arrival time values): thus each of these operations requires constant time.
The algorithm performs $|\mathbb{E}|$ iterations, one for each temporal edge $e=\tedge$. For each iteration, the only operation that does not require constant time is finding in $\poset s{t_0}u$ the pair $(t^u,c^u)$ with greatest arrival time $t^u\leq\tau$ (see line~\ref{alg:costly} of Algorithm~\ref{alg:paretoproblem}). We will now give a bound on the time complexity of this operation. From the definition of the Pareto complexity of a temporal edge, we have that $(t^u,c^u)$ is located exactly before the last $K_e$ elements of the current $\poset s{t_0}u$. We can then look for it in the following way. Let $p$ be the size of the current $\poset s{t_0}u$. We then look at the pairs in position $p$, $p-1$, $p-2$, $p-4$, $p-8$, \dots, until we find, in a certain position $p-k$ (after $O(\log k)$ steps), a pair with arrival time less than or equal to $\tau$. The pair we are looking for is now in a position between $p-k$ and $p-k/2$: by using a binary search technique, we can find it in $O(\log k)$ iterations. Since $k\leq 2K_e$, the total cost to perform these operations is $O(\log K_e)$. We can then conclude that the time complexity of the algorithm is $O(|\mathbb{E}|\log K)$.
For what concerns the space complexity, during each iteration the algorithm adds at most two pairs to a Pareto set if the tail of the scanned temporal edge is the source $s$, and at most one pair in the other cases. This guarantees that the total number of pairs in the Pareto sets is bounded by $2|\mathbb{E}|$. Thus the algorithm executes in space $O(|\mathbb{E}|)$, and the theorem follows.
\end{proof}

\subsection{Extending the algorithm to multiple cost structures}\label{sec:multicost}

Given a PATDN $N=(V,\mathbb{E})$, let us consider $h$ cost structures $\mathcal{C}_1,\dots,\mathcal{C}_h$ over $\mathbb{E}$, with each $\mathcal{C}_i=(\cost_i,\gamma_i,\oplus_i,\preceq_i)$ satisfying the isotonicity property. Let $\cost = \cost_1\times\cdots\times\cost_h$. We define a global cost function $\gamma:\mathbb{E}\rightarrow\cost$ as follows: for each $e\in\mathbb{E}$, $\gamma(e)=(\gamma_1(e),\ldots,\gamma_h(e))$. We also define a global cost combination function $\oplus : \cost\times\cost\longrightarrow\cost$ as follows: given $c_1=(c_{1,1}, \dots,c_{1,h}), c_2=(c_{2,1},\dots,c_{2,h})\in\cost$, $c_1 \oplus c_2 = (c_{1,1} \oplus_1 c_{2,1}, c_{1,2} \oplus_2 c_{2,2}, \dots, c_{1,h} \oplus_h c_{2,h})$. Finally, we define a global cost total lexicographical order $\preceq\subseteq\cost\times\cost$ as follows: given $c_1=(c_{1,1}, \dots,c_{1,h}), c_2=(c_{2,1},\dots,c_{2,h})\in\cost$, $c_1 \preceq c_2$ if and only if $c_1 = c_2$ or $c_{1,i}\prec_i c_{2,i}$ for the first index $i$ such that $c_{1,i}\neq c_{2,i}$, where $\prec_i$ is the strict order on $\cost_i$ induced by $\preceq_i$. It is easy to verify that the cost structures $\mathcal{C}=(\cost,\gamma,\oplus,\preceq)$ over $\mathbb{E}$ satisfies the isotonicity property. The global cost of a path is defined similarly to what we have done with one cost function. The definition of multi-criteria Pareto $t_0$-optimal paths and of the Pareto problem are also similar. We can now use Algorithm~\ref{alg:paretoproblem} to solve the multi-criteria Pareto problem. In this way we manage to combine a plurality of costs and prioritize paths with respect to different criteria, provided that an order of importance between them is given. Note that this approach is  different from looking for the Pareto sets with respect to $h$ different costs of equal importance.

\subsection{Relaxing the constraint on temporal edges with zero delay}
\label{subsec:zero}

In the description of Algorithm~\ref{alg:paretoproblem} we have assumed that the zero $t$-snapshot graph $G_t$ is a directed acyclic graph. In this section, we show how this hypothesis can be further relaxed in order to deal with more general cases, whenever the $\oplus$ operation is associative. Let $N=(V,\mathbb{E})$ be a PATDN with a cost structure $\mathcal{C}=(\cost,\gamma,\oplus\preceq)$ over $\mathbb{E}$, and let us consider a weighted version of the graph $G_t$ in which the weight $w(e)\in\cost$ of an edge $e=(u,v)$ is equal to $\gamma((u,v,t,0))$. We say that a strongly connected component $C$ of $G_t$ is \emph{cost transitively closed} if, for each pair of edges $e_1=(u_1,u_2)$ and $e_2=(u_2,u_3)$ in $C$, there always exists an edge $e_3=(u_1,u_3)$ in $C$ with $w(e_3)\preceq w(e_1)\oplus w(e_2)$. The PATDN $N$ is said to be \emph{zero transitively closed} if, for each $t\in\thor$, all strongly connected components of $G_t$ are cost transitively closed. Algorithm~\ref{alg:paretoproblem} can then be adapted to take as input any zero transitively closed PATDN by refining the order in which the temporal edges in $E$ are scanned: we additionally require that the edges with the same arrival time $t$ and having delay zero, i.e. those that correspond to edges in $G_t$, are given according to a topological order of the strongly connected components of $G_t$. The correctness follows from the fact that the cost transitively closed hypothesis allows us to consider only paths with at most one edge in any strongly connected component of any $G_t$, as any other path is dominated by such a path.


\subsection{Finding Pareto optimal paths}

Algorithm~\ref{alg:paretoproblem} computes the value of the arrival time and of the cost of the Pareto $t_0$-optimal paths. In this section we describe how to compute, for each pair of these values, a corresponding Pareto $t_0$-optimal path. We proceed as follows. To each value $(t,c)$ in a Pareto set $\poset s{t_0}v$, we will associate two pointers during the execution of the algorithm. Consider the iteration of the algorithm during which $(t,c)$ is added to $\poset s{t_0}v$, and let $e=\tedge$ be the edge which caused this update. We then associate to $(t,c)$ a pointer $\pi_1$ to $e$. Concerning the second pointer $\pi_2$, if $(t,c)$ is added to $\poset s{t_0}v$ at line \ref{alg:update2}, we associate to $(t,c)$ a pointer to $(t^u,c^u)\in\poset s{t_0}u$, otherwise we associate to $(t,c)$ a null pointer. In order to extract a path $\tpath P$ corresponding to a pair $(t,c)\in\poset s{t_0}v$, we can recursively proceed backwards, from its last edge to the first one, as follows:
\[
p(x,y) = \left\{\begin{array}{ll}
p(\pi_2(x,y)) \mbox{ concatenated with $\pi_1(x,y)$} & \mbox{if $\pi_2(x,y)$ is not null,}\\
\langle\pi_1(x,y)\rangle & \mbox{otherwise.}\\
\end{array}\right.
\]
The Pareto $t_0$-optimal path $\tpath P$ corresponding to the pair $(t,c)\in\poset s{t_0}v$ can then be computed as $p(t,c)$.

\section{Applications}
\label{sec:applications}

In this chapter we will show several possible applications of Algorithm~\ref{alg:paretoproblem}. To this aim, for each application, we will specify the cost structure to be used: it is easy to verify that each cost structure satisfies the isotonicity property. In the following, we assume that $t_0$ is any fixed time instant in $\thor$, and that the paths are all starting no earlier than $t_0$.

\noindent\textbf{Profile problem:} compute, for each destination $u$, its profile function, which associates, to any starting time $t$, the earliest arrival time in $u$, if we start from $s$ at time $t$. In order to solve this problem, we use the following cost structure: (a) $\cost = \mathbb{R}$, (b) for each temporal edge $e = (u,v,\tau,\delta)$, $\gamma (e) = \tau$, (c) for any two real numbers $a$ and $b$, $a\oplus b = a$, and (d) for any two real numbers $a$ and $b$, $a\preceq b$ if and only if $a \geq b$. Note that according to this cost structure, the cost of a path is equal to its starting time. The Pareto set $\poset s{t_0}u$ allows us to compute the profile function of $u$ because of the following reason. If we consider two consecutive pairs $(a_1,s_1)$ and $(a_2,s_2)$ in $\poset s{t_0}u$, we have that $a_1 < a_2$ and $s_1 < s_2$ because of the Pareto optimality. Hence, for any starting time $t \in (s_1, s_2]$, we can deduce that the earliest arrival time in $u$ is $a_2$, since $\poset s{t_0}u$ contains all Pareto optimal pairs and no pair can have departure time in-between $s_1$ and $s_2$.
    
\noindent\textbf{Fewest hops:} compute, for each destination $u$, the minimum number of edges of a path from $s$ to $u$. In order to solve this problem, we use the following cost structure: (a) $\cost = \mathbb{N}$, (b) for each temporal edge $e$, $\gamma (e) = 1$, (c) for any two natural numbers $a$ and $b$, $a \oplus b = a+b$, and (d) for any two natural numbers $a$ and $b$, $a\preceq b$ if and only if $a \leq b$. Note that according to this cost structure, the cost of a path $\tpath P = \langle e_1,e_2, \dots, e_k \rangle$ is equal to $k$. Hence, to obtain the minimum number of edges needed by a path to reach a node $u$ from the source $s$, it suffices to look at the cost of the last pair in $\poset s{t_0}u$.
	
\noindent\textbf{Shortest delay:} compute, for each destination $u$, the minimum delay of a path from $s$ to $u$, where the delay of a path is the sum of the delays of its temporal edges. In order to solve this problem, we use the following cost structure: (a) $\cost = \mathbb{R}^+$, (b) for each temporal edge $e = (u,v,\tau,\delta)$, $\gamma (e) = \delta$, (c) for any two real numbers $a$ and $b$, $a \oplus b = a+b$, and (d) for any two real numbers $a$ and $b$, $a\preceq b$ if and only if $a \leq b$. Note that according to this cost structure, the cost of a path is the sum of the delays of its temporal edges. Hence, to obtain the cost of the shortest delay path from the source $s$ to each node $u$, it suffices to look at the cost of the last pair in $\poset s{t_0}u$.
	
\noindent\textbf{Shortest fastest:} compute, for each destination $u$, the minimum number of edges of a path from $s$ to $u$, among all the paths with minimal duration, where the duration of a path is defined as the difference between its arriving and starting times. In order to solve this problem, we use the following two cost structures: (a) $\cost_1 = \mathbb{R}$ and $\cost_2 = \mathbb{N}$, (b) for each temporal edge $e = (u,v,\tau,\delta)$, $\gamma_1(e) = \tau$ and $\gamma_2(e) = 1$, (c) for any two real numbers $a$ and $b$, $a\oplus_1 b = a$ and $a \oplus_2 b = a+b$, and (d) for any two real numbers $a$ and $b$, $a\preceq_1 b$ if and only if $a \geq b$ and $a\preceq_2 b$ if and only if $a \leq b$. Note that according to cost structure $\cost_1$, the cost of a path is 
its starting time (latest being preferred).
Combining these cost structures $\cost_1$ and $\cost_2$ as explained in \ref{sec:multicost} allows us to compute the values of the Pareto optimal paths with respect to arrival time and a cost that has the departure time in the first component and the number of hops in the second one. For each Pareto set, it hence suffices to extract the pair $(a, (d,h))$ such that $a-d$ is minimal among all the the pairs in the Pareto set. 

\noindent\textbf{MaxProd${}^+$:} compute, for each destination $u$, the maximum value of a path from $s$ to $u$, where the value of a path is defined as the product of the costs of its temporal edges. In order to solve this problem, we use the following cost structure: (a) $\cost = \mathbb{R}^+$, (b) $\gamma$ can be any function, (c) for any two real numbers $a$ and $b$, $a \oplus b = a\cdot b$, and (d) for any two real numbers $a$ and $b$, $a\preceq b$ if and only if $b \geq a$. Note that according to this cost structure, the cost of a path is the product of the costs of its temporal edges. Hence, to obtain the cost of the maximum value path from the source $s$ to each node $u$, it suffices to look at the cost of the last pair in $\poset s{t_0}u$.
	
\noindent\textbf{MinMax:} compute, for each destination $u$, the minimum requirement of a path from $s$ to $u$, where the requirement of a path is defined as the maximum of the costs of its temporal edges.  In order to solve this problem, we use the following cost structure: (a) $\cost = \mathbb{R}$, (b) $\gamma$ can be any function, (c) for any two real numbers $a$ and $b$, $a \oplus b = \max(a,b)$, and (d) for any two real numbers $a$ and $b$, $a\preceq b$ if and only if $a \leq b$. Note that according to this cost structure, the cost of a path is the maximum of the costs of its temporal edges. Hence, to obtain the cost of the minimum value path from the source $s$ to each node $u$, it suffices to look at the cost of the last pair in $\poset s{t_0}u$. Note that it is also possible to solve, in an analogous way, any combination of the minimum or maximum selection with the sum, product, minimum, or maximum composition.

\section{Conclusion and open questions}
\label{sec:conclusion}

We have described and analysed a general algorithm for solving the Pareto problem in PATDN, which significantly improves the time complexity of the previously known solution, and which can be used to solve several different minimum cost path problems in PATDN with a vast variety of cost definitions. Even if the Pareto problem we considered is defined as a one-to-all path problem, our algorithm can be easily adapted in order to deal with all-to-one path problems. This can be obtained by making the latest starting time play the role of the earliest arrival time, by basically scanning the edges in reverse order, and by assuming a symmetric version of the isotonicity property. As we already said in the introduction, our algorithm adapts to undirected TDNs with no temporal edges with zero delay.

It would be interesting to consider the complexity of the Pareto problem in the case of TDNs at higher levels of the hierarchy described in the introduction. Moreover, in our formulation of the Pareto problem the starting time is fixed, and it would be interesting to consider the problem of computing the Pareto optimal function, returning for each time instant the corresponding Pareto sets, by providing a solution faster than the obvious one consisting of applying our algorithm for each possible time instant. Finally, we think that it is worth exploring different version of the Pareto problem in which the arrival time is substituted by some other criteria, such as, for example, the duration of a path.

\bibliography{biblio}

\newpage

\appendix

\end{document}